\documentclass[12pt, english, a4paper,reqno]{amsart}

\usepackage{graphicx}
\usepackage{tikz}
\usetikzlibrary{positioning,arrows}
\usetikzlibrary{arrows.meta}

\usepackage{capt-of}
\usepackage[applemac]{inputenc}
\usepackage{hyperref}
\usepackage{graphicx}

\usepackage{lscape}
\usepackage{amsmath,amssymb}
\usepackage{bm}
\usepackage{euscript,color}
\usepackage{babel}
\usepackage{amsthm}
\usepackage{amsfonts}
\usepackage{latexsym}
\usepackage{fancyhdr}
\usepackage{enumerate}

\usepackage{dsfont}
\usepackage{color}
\usepackage{mathtools}

\usepackage[numbers]{natbib}

\usepackage[left=2.5cm, top=3cm, bottom=3cm, right=2.5cm]{geometry}
\usepackage{soul}

\usepackage{setspace}
\usepackage{cancel}

\newcommand{\R}{\mathbb{R}}

\newcommand{\Chi}{\mathfrak{X}}

\newcommand{\comment}[1]{}

\makeatletter
\renewcommand{\section}{\@startsection%
{section}
{1}
{0mm}
{1.5\bigskipamount}
{0.5\bigskipamount}
{\centering\normalsize\sc}}

\renewcommand{\paragraph}{\@startsection%
{paragraph}
{4}
{0mm}
{\bigskipamount}
{-1.25ex}
{\normalsize\sl}}

\def\provedboxcontents#1{$\square$}

\makeatother

\newtheorem{newstatement}{newstatement}

\newtheorem{theorem}[newstatement]{Theorem}
\newtheorem{corollary}[newstatement]{Corollary}

\newtheorem*{conjecture*}{Conjecture}
\newtheorem*{prop*}{Proposition}

\theoremstyle{def}

\theoremstyle{rmk}
\newtheorem{remark}[newstatement]{Remark}

\theoremstyle{claim}

\theoremstyle{remark}

\expandafter\let\expandafter\oldproof\csname\string\proof\endcsname
\let\oldendproof\endproof
\renewenvironment{proof}[1][\proofname]{%
  \oldproof[\slshape #1]%
}{\oldendproof}

\let\phi\varphi
\let\epsilon\varepsilon

\renewcommand{\emph}[1]{{\slshape #1}}
\renewcommand{\em}{\sl}

\title[Uniqueness and redistributive policies]{Uniqueness of Equilibrium and redistributive policies: A Geometric Approach to Efficiency}

\author{Andrea Loi}
\address{Andrea Loi, Dipartimento di Matematica e Informatica \\
         Universit\`a di Cagliari, Italy.}
         \email{loi@unica.it}

\author{Stefano Matta}
\address{Stefano Matta, Dipartimento di Scienze economiche e Aziendali \\
         Universit\`a di Cagliari, Italy.}
         \email{smatta@unica.it}
         
\author{Daria Uccheddu}
\address{Daria Uccheddu, Dipartimento di Matematica e Informatica\\
         Universit\`a di Cagliari, Italy.}
         \email{daria.uccheddu@unica.it}

\date{\today}

\thanks{The first  and third author were  supported  by INdAM. GNSAGA - Gruppo Nazionale per le Strutture Algebriche, Geometriche e le loro Applicazioni and
by GOACT by Fondazione di Sardegna. The first author was also  supported by PNRR e.INS Ecosystem of Innovation for Next Generation Sardinia (CUP F53C22000430001, codice MUR ECS00000038}

\begin{document}
\begin{abstract}
This paper examines the relationship between resource reallocation, uniqueness of equilibrium and efficiency
in economics.
We explore the implications of reallocation policies for stability, conflict, and decision-making by analysing the existence of geodesic coordinate functions in the equilibrium manifold. 

 Our main result shows that in an economy with $M = 2$ consumers and $L$ goods, 
if $L$ coordinate functions, representing policies,  are geodesics on the equilibrium manifold (a property that we call the {\em finite geodesic property}), then the equilibrium is globally unique. The presence of geodesic variables indicates optimization and efficiency in the economy, while non-geodesic variables add complexity. 

Finally,  we establish a link between the existing results on curvature, minimal entropy, geodesics and uniqueness in smooth exchange economies.

This study contributes to the understanding of the geometric and economic properties of equilibria and offers potential applications in policy considerations.


{\it{Subj.Class}}: {C61, C65, D50, D51.}

\end{abstract}
\maketitle

\vspace{0.3in}

\section{Introduction}\label{introduction}

In the absence of uniqueness of equilibrium, due to the role of belief formation and expectation coordination in equilibrium selection, redistributive policies may affect efficiency and social stability, by increasing the likelihood of conflict and unpredictable change and by not allowing decision-makers to make informed choices based on a well-defined equilibrium point, which in turn may lead recursively to more unstable outcomes.

The motivation of this paper is to establish a link between the resource redistribution and uniqueness. Under what conditions can redistributive policies take place under uniqueness? In our very general setting, a redistributive policy can be thought of as a change of one or more economic variables across the different states of the system representing the economy. It can therefore be represented by a smooth curve on the equilibrium manifold,  which is a mathematical representation of all possible states of an economic system in equilibrium. Each point corresponds to an equilibrium state where all relevant economic variables are in balance, and there is no tendency for the system to change unless external factors intervene.

Since a policy should be efficient, it is a natural assumption to consider that it minimises an objective function, which can be represented by distance with an appropriate metric.

A geodesic, roughly speaking, represents a straight line in a curved space or, more precisely, a curve on a manifold that minimises distance locally. Its behaviour is deeply related to the curvature of the manifold and its metric, object of analysis in differential geometry. This concept is ubiquitous, ranging from 
theoretical physics to games, information geometry, machine learning, operations research, to name but a few.

This branch of mathematics has been applied to statistics and information theory  \cite{ama, muri} 
and subsequently to econometrics (see e.g. \cite{masa}). However, very few results can be found in the economic literature, despite of the fact that the minimisation property of the geodesic fits naturally into the economic framework of optimality and efficiency and, being closely related to the metrics of the space under analysis, allows for the development of powerful methodologies that take into account several aspects considered relevant. 
As an example, \cite{kri} studies the existence of Nash equilibria by embedding the non-convex domains of the payoff functions in suitable Riemannian manifolds.
In general equilibrium theory, \cite{lmgeo, lmcatmin} use geodesics to minimise the cost of redistribution policies and catastrophes. 
But, more generally, lack of use of geodesics had been observed by Balasko \cite{balib}  with regard to planning theory and general equilibrium theory, where, on the other hand, the differential topology approach, following the ground-breaking article by \cite{deb}, has been prevalent and well suited to deal with global properties, multiplicity, and uniqueness.

Our main result is Theorem \ref{mainteor}, where we show that, in an economy with $M=2$ consumers and $L$ goods, if $L$  coordinate functions are geodesics on the equilibrium manifold (a property that we call {\em finite geodesic property}), then the equilibrium is globally unique.

Since each coordinate function acts on an economic good, the intuition behind the result is that if all the economic resources follow optimal paths in the equilibrium manifold, the equilibrium is globally unique. The interpretation in terms of optimisation and efficiency is that when all the economic variables are optimised and respond
efficiently to changes in the economic environment (e.g., policy distributions), this can lead to uniqueness, 
and hence more efficient resource allocation, cost minimisation and welfare maximisation in the economy.

The result lends itself to different interpretations and potential applications.
The existence of non-geodesic coordinate functions indicate that some economic resources do not necessarily behave optimally in the equilibrium manifold. These variables may respond less efficiently to changes in the economic environment, leading to potential inefficiencies or suboptimal outcomes. This can lead to new approaches in partial equilibrium analysis, through a further analysis of specific sectors and industries or market segments where optimization is more or less likely. Moreover, this perspective can also pave the way for a different approach to complex economic dynamics: the coexistence of geodesic and non-geodesic coordinate functions implies that the economic system is not entirely homogenous in terms of optimising behaviour. 
This can lead to interesting nonlinear dynamics, where different sectors or variables may interact in different ways and respond differently to shocks or policy changes.
Finally, understanding which variables do and do not follow geodesic paths and which do not can have important policy implications. Policymakers can focus on optimising those variables that do not follow geodesic paths in order to improve overall economic efficiency and stability.

In summary, the presence of a finite number of geodesic coordinate functions in the equilibrium manifold indicates partial optimization and efficiency in certain aspects of the economy. It also highlights areas where improvements can be made to achieve better overall economic outcomes. However, the coexistence of non-geodesic variables adds complexity to the economic dynamics and requires careful analysis and policy considerations.

The technique employed is efficient, limiting the number of curves to be studied. Furthermore, although the work is theoretical, the considerable development of algorithms that allow the geodesics of even implicit functions to be  determined numerically makes this work potentially useful for applications. 

 Finally, we believe that this result is interesting because it deepens our knowledge of the connections between geometric and economic properties of equilibria. In fact, it is deeply connected to the recent results on uniqueness using the zero curvature \cite{lmcurvuni,lmu} and minimal entropy \cite{lment} conditions.
 Corollary \ref{corequiv} summarises, to the best of our knowledge, the state of the art on these relationships in smooth exchange economies.

In recent years the concepts of differential geometry have had several applications in discrete spaces with economic properties. For example, following a network approach,
curvature has been associated with system fragility (see, e.g., \cite{jost,sgt}).
There are still relatively few papers in the economic literature that use differential geometry in its most natural domain, that of manifolds. The aforementioned works highlight a link between curvature, entropy and uniqueness.
We hope that the present work will further stimulate the reader's interest in making connections with results obtained in discrete spaces and exploring metrics with economic significance.

This paper is organized as follows. Section \ref{setting} recalls the main properties of the equilibrium manifold. Section \ref{mathprel} presents the main concepts of differential geometry used to prove our main result. Section \ref{basic} analyses the consequences of the finite geodesic property in low dimension, when the equilibrium manifold is a surface. Finally, in Section \ref{mainsection}
we prove our main result, Theorem \ref{mainteor}, and explain in Corollary \ref{corequiv} the connections between the finite geodesic property and curvature, minimal entropy and uniqueness on smooth economies.

\section{The economic setting}\label{setting}

We consider a smooth pure exchange economy with $L$ goods and $M$ consumers. 
The equilibrium manifold $E$ is defined as the set
of pairs $(p,\omega)$ such that the excess demand function is zero, where
$p$ belongs to the set of normalized prices $S=\{p=(p_1,\ldots,p_L)\in \R^L| p_l>0,l=1,\ldots,L,p_L=1\}\cong \mathbb{R}^{L-1}$ and $\omega$ belongs to the space of endowments $\Omega=\R^{LM}$.
The equilibrium manifold enjoys very nice geometric properties, being a smooth submanifold of $S\times\Omega$ globally diffeomorphic to $\R^{LM}$ \citep[Lemma 3.2.1]{balib}. 

If total resources $r\in\R^L$ are fixed,  the equilibrium manifold, denoted by $E(r)$, is a submanifold of $S\times\Omega(r)$ globally diffeomorphic to $\R^{L(M-1)}$ \citep[Corollary 5.2.5]{balib}, that is, $E(r)\cong B(r)\times \R^{(L-1)(M-1)}$, where $B(r)$ denotes the {\em price-income equilibria}, a submanifold of $S\times \R^M$ diffeomorphic to $\R^{M-1}$ \citep[Corollary 5.2.4]{balib}.
If we define this diffeomorphism as

\begin{equation}\label{parB}
\phi:\R^{M-1}\to B(r),
\end{equation}

$$t=(t_1,\ldots,t_{M-1})\mapsto  (p(t),w_1(t),\ldots,w_{M-1}(t)),$$

where $w_i$ denotes consumer $i$'s income, a parametrization of $E(r)$ 
(see formulas (6), (7) and (10) in \cite{lmcurvuni}) is given by 

\begin{equation}\label{parE}
\Phi:\R^{L(M-1)}\to E(r),
\end{equation}

$$(t,\bar\omega_1,\ldots,\bar\omega_{M-1})\mapsto (p(t),\bar\omega_1,w_1(t)-p(t)\bar\omega_1,\ldots,
\bar\omega_{M-1},w_{M-1}(t)-p(t)\bar\omega_{M-1}),$$
where $\bar\omega_i$ denotes the first $L-1$ components of $\omega_i$,
consumer $i$'s endowments vector.

\section{Mathematical Preliminaries}\label{mathprel}

In this section we present some mathematical definitions that will be used later.
We refer readers unfamiliar with differential geometry to, e.g., \cite{docs,do}.
Consider a submanifold $M$ of $(\mathbb R^n, g_{euclid})$ of dimension $m$,
where $g_{euclid}$ denotes the Euclidean \lq\lq flat'' metric. A metric is naturally induced on $M$ by its ambient space as follows. If
\[
\begin{split}
\phi :  &\: \mathbb{R}^m\longrightarrow M\subset \mathbb R^n,\\
&(x_1,\dots,x_m)\mapsto (\phi_1,\dots\phi_n),
\end{split}\] 
is a parametrization of $M$, we can consider the vector fields  
$X_1=(\frac{\partial \phi_1}{\partial x_1},\dots,\frac{\partial \phi_1}{\partial x_m})$, $X_2=(\frac{\partial \phi_2}{\partial x_1},\dots,\frac{\partial \phi_2}{\partial x_m})$ and  $X_n=(\frac{\partial \phi_n}{\partial x_1},\dots,\frac{\partial \phi_n}{\partial x_m})$. Moreover,  $\{X_1,\dots, X_n\}$ is a basis of vector fields of $T_pM$ for $p\in M$. The ambient space $\mathbb R^n$ induces on $M$ a metric given by $$ds^2=\displaystyle\sum_{i,j=1}^{m}g_{ij}dx_idx_j,$$ where $g_{ij}=\langle X_i,X_j\rangle_{g_{euclid}}$. \\
Let us denote with $\Chi (M)$ the set of all vector fields of class $C^{\infty}$ on $M$. There exists an affine connection\footnote{One can think of the connection  as an operator that enables to apply differentiation on an abstract manifold $M$ intrinsically, that is without referring to its ambient space, by extending the concept of directional derivative to vector fields defined on $M$.}
 \[\begin{split}
 \nabla: \Chi(M)\times& \Chi(M)\rightarrow \Chi(M)\\
 (X,&Y)\mapsto \nabla_XY
 \end{split}\] which satisfies the following properties:
 \begin{itemize}
 \item $\nabla_{fX+gY}=f\nabla _XZ+g\nabla_YZ$
 \item $\nabla_X(Y+Z)=\nabla_XY+\nabla_XZ$
 \item $\nabla_X(f Y)=f\nabla_X Y+X(f)Y$
 \end{itemize} 
 where $f$ and $g$ are real-valued functions of class $C^{\infty}$ on $M$.
 
The {\em Lie-braket}, defined as $[X,Y]=XY-YX$, is a vector field belonging to $\Chi(M)$.

\vspace{0.5 cm}

{\bf Theorem} (Levi-Civita see p. 55 in \cite{do}) {\em Given a Riemannian manifold $M$, there exists a unique affine connection $\nabla$ on $M$ satisfying the conditions:
\begin{itemize}
\item $\nabla $ is symmetric, that is $[X,Y]=\nabla_XY-\nabla_YX$;
\item $\nabla$ is compatible with the Riemannian metric:  $X\langle Y,Z\rangle=\langle \nabla_XY,Z\rangle+ \langle Y,\nabla_XZ\rangle$ for $X, Y, Z\in \Chi (M)$.
\end{itemize}}

\vspace{0.5 cm}

In particular, the Levi-Civita connection can be written, in a coordinate system $(U,x)$ as 
$$\displaystyle\nabla_{\frac{\partial}{\partial x_i}}\frac{\partial}{\partial x_j}=\displaystyle\sum_{k=1}^{m}\Gamma_{ij}^k\frac{\partial}{\partial x_k},$$ 
where the coefficients $\Gamma_{ij}^k$ are called the Christoffell symbols and can be computed using the following formula
$$\Gamma_{ij}^k=\frac{1}{2}\displaystyle\sum_{h=1}^{m}g^{hk}\left( \frac{\partial g_{jh}}{\partial x_i} + \frac{\partial g_{hi}}{\partial x_j}-\frac{\partial g_{ij}}{\partial x_h}\right),$$
where $g^{ij}$ is the inverse matrix of $g_{ij}=\langle X_i,X_j\rangle$.\\

A parametrized curve $\gamma:I\to M$ is a {\em geodesic} if $\nabla_{\dot \gamma}\dot \gamma=0$. This curve can be seen a generalization of the concept of straight line, whose velocity vector has derivative equal to zero.

The following proposition  enables to study the geodesic on abstract Riemannian manifolds in terms of ordinary differential equations.
\vspace{0.5cm}

\noindent {\bf Proposition} (see \cite{do}).

\noindent {\em Let $M$ be a manifold endowed with a Levi-Civita connection $\nabla$. Choose coordinates $x^i$ on some neighborhood $U$ of $p$. For any $p\in M$, any $V\in T_pM$, and any $t_0\in \mathbb R$, there exists an open interval $ I\subset \mathbb R$ containing $t_0$ and a geodesic $\gamma:I\rightarrow M$ such that $\gamma(t_0)=p$, and $\dot \gamma (t_0)=V.$ If $\gamma(t)=\left(x^1(t),\dots,x^n(t)\right)$, then its components satisfy \begin{equation}\label{geodesic}\ddot x^k(t)+\Gamma^k_{ij}x(t)\dot x^i(t)\dot x^j(t)=0.\end{equation}
Equation \eqref{geodesic} is an ordinary differential equation with initial condition $\gamma(t_0)=p$ and $\dot \gamma (t_0)=V$, hence, by the fundamental theorem of differential equations, it has
an unique solution.}

\vspace{0.5 cm}

Finally, we observe that a geodesic can also be characterized  by the property of having its acceleration vector in $\mathbb{R}^n$ everywhere parallel to the  normal vector to the surface, see e.g. {\citep[p.246]{docs}. 

\section{Basic Example with $L=2$ and $M=2$}\label{basic}

In this section, we show that, with two consumers and two commodities, if the coordinate curves $\Phi(t,0)$ and $\Phi(0,t)$ are geodesics, then the price is constant.
This is a special case we
introduce here to provide the intuition before presenting  in the next section
the proof of the more general case with an arbitrary number of goods.

Observe that if $L=M=2$ the equilibrium manifold $E(r)$ is a surface, i.e. a submanifold of dimension $2$ in a 
three-dimensional space \footnote{  $E(r)\subset S\times \Omega(r)$, where $\Omega(r)=\{\omega=(\omega_1,\omega_2), \,\vert\, \omega_i\in \mathbb R^2, \omega_1+\omega_2=r\}=\{(\omega_1,r-\omega_1)\}\cong \mathbb R^{2}$.}. An explicit parametrization of $E(r)$ is given by
\citep[Section 4]{lmcurvuni}

$$ \Phi: \mathbb{R}^2  \to E(r), $$ where $(t,\alpha) \mapsto (p(t),\alpha,w(t)-p(t)\alpha).$

Consider a basis of vector fields in $T_xE(r)$ given by
$$\Phi_0=\left(\frac{\partial p}{\partial t},0,\frac{\partial w}{\partial t}-\frac{\partial p}{\partial t}\alpha\right)=(\dot p,0,A)$$
$$\Phi_{1}=\left(0,1,-p\right),$$
where
 $A:= \dot w-\dot p \:\alpha$.

\noindent Observe that $E(r)$ is a ruled surface (see \citep[p.188]{docs}. In fact, 
$\Phi(t,v)=\Phi(t,0)+v\:\Phi_1(t)$. The normal vector to the surface is the orthogonal vector to the tangent space $T_xE(r)=\mbox{span}\{(\dot p,0,\dot w-\dot p\alpha),(0,1,-p)\}$, that is $N=(\dot p\alpha-\dot w, p\dot p,\dot p)$.

In particular, taking into account the last observation of Section \ref{mathprel}, if we suppose that the coordinate curve $\Phi(t,0)$ is a geodesic, then the second derivative of the parametrization $\ddot \Phi$ is parallel to the normal vector $N$ to the surface, that is  $\ddot \Phi\wedge N=(0,0,0)$.

An explicit computation of $\ddot \Phi(t,0)$ gives
$\ddot\Phi(t,0)=(\ddot p(t),0,\ddot w(t))$ and hence, from $\ddot \Phi\wedge N=(-\dot w p \dot p, \ddot w(\dot p\alpha-\dot w)-\dot p\ddot p,p \dot p\ddot p)$,
we have that this vector is the null vector if and only if 
\[
\left\{
\begin{array}{rrr}
-\dot w p \dot p=0\\
\ddot w(\dot p\alpha-\dot w)-\dot p\ddot p=0\\
p \dot p\ddot p=0
\end{array}
\right.
\]

Recalling that the price $p$ is different from zero, from the last equation we have that $\ddot p=0$
and then $\dot p=b$ and $p(t)=b\:t+c$. Since the price is positive, that is $p(t)>0$, the only possibility is that $p(t)$ is a constant line always positive, that is a line parallel to the $t$ axis, and so the price $p$ is constant. Clearly, the other coordinate curve, $\Phi(0,t)$, is always a geodesic because we have
$\ddot\Phi(0,t)=(0,0,0)$. \\

Now we present the general idea of the proof that will be used in the more general case, using local coordinates through the Christoffel symbols.\footnote{For an explicit and tedious calculation of the Christoffel symbols, we refer the interested reader to \cite{lmu}.}
We are going to prove that the price is constant if the coordinate curves $\Phi(t,0)$ and $\Phi(0,t)$ are solutions of the ODE  \eqref{geodesic}. Curves will be parametrized by arc length.

\noindent The Christoffel symbols are given by
$$\Gamma_{ij}^k=\frac{1}{2}\displaystyle\sum_{h=0}^{1}g^{hk}\left( \frac{\partial g_{jh}}{\partial x_i} + \frac{\partial g_{hi}}{\partial x_j}-\frac{\partial g_{ij}}{\partial x_h}\right).$$
In particular, with our parametrization, they are (see \cite[Appendix A]{lmu})
\[
\begin{split}
\Gamma_{00}^0=&\frac{[(1+p^2)\dot p  \ddot p+(\dot w-\dot p x_1)(\ddot w-\ddot px_1)]}{(1+p^2)(\dot p)^2+A^2},\\
\Gamma_{00}^1=&\frac{p\left[(\dot w-\dot p x_1)\dot p\ddot p-(\ddot w-\ddot px_1)(\dot p^2)\right]}{(1+p^2)(\dot p)^2+A^2},\\
\Gamma_{10}^0=&\frac{-\dot p(\dot w-\dot p x_1)}{(1+p^2)(\dot p)^2+A^2},\\
\Gamma_{10}^1=&\frac{\dot pp (\dot p)^2}{(1+p^2)(\dot p)^2+A^2},\\
\Gamma^0_{11}=&\Gamma^1_{11}=0.
\end{split}\]

\noindent If $\gamma(t)=\left(\gamma_0,\gamma_1\right)\subset E(r)=\Phi(t,\alpha)$,
Equation \eqref{geodesic}
becomes
\[
\left\{
\begin{array}{c}
\ddot \gamma_0(t)+\Gamma^0_{00}\dot \gamma_0(t)^2+2\,\Gamma^0_{01}\dot \gamma_0(t)\dot \gamma_1(t)=0\\
\ddot \gamma_1(t)+\Gamma^1_{00}\dot \gamma_0(t)^2+2\,\Gamma^1_{01}\dot \gamma_0(t)\dot \gamma_1(t)=0,
\end{array}
\right.
\]
which can be written, computing the Christoffel symbols,

\begin{equation}\label{sist}
\left\{
\begin{array}{crrr}
\ddot \gamma_0+\dfrac{[(1+p^2)\dot p\ddot p+(\dot w-\dot p\gamma_1)(\ddot w-\ddot p\gamma_1)]}{(1+p^2)(\dot p)^2+A^2}\dot \gamma_0^2+2\dfrac{-\dot p(\dot w-\dot p\gamma_1)}{(1+p^2)(p')^2+A^2})\dot \gamma_0\dot \gamma_1=0\\
\:\\
\ddot \gamma_1+\dfrac{p\left[(\dot w-\dot p\gamma_1)\dot p\ddot p-(\ddot w-\ddot p\gamma_1)(\dot p^2)\right]}{(1+p^2)(\dot p)^2+A^2}\dot \gamma_0^2+2\dfrac{\dot pp (\dot p)^2}{(1+p^2)(\dot p)^2+A^2}\dot \gamma_0\dot \gamma_1=0.
\end{array}
\right.
\end{equation}
In particular, if the curve is given by $\Phi(t,0)$, we have $ \gamma_1=0$ and then we can write
\[
\left\{
\begin{array}{crr}
\ddot \gamma_0+\dfrac{[(1+p^2)\dot p\ddot p+\dot w\ddot w]}{(1+p^2)(\dot p)^2+A^2}\dot \gamma_0^2=0\\
\:\\
\dfrac{p\left[\dot w\dot p \ddot p-\ddot w(\dot p^2)\right]}{(1+p^2)(\dot p)^2+A^2}\dot \gamma_0^2=0.
\end{array}
\right.
\]
If $\gamma_0=t$ then $\dot \gamma_0=1$ and  $\ddot \gamma_0=0$, and hence we obtain

\[
\left\{
\begin{array}{c}
[(1+p^2)\dot p\ddot p+\dot w\ddot w]=0\\
p\left[\dot w\dot p \ddot p-\ddot w(\dot p^2)\right]=0
\end{array}
\right.
\]
Moreover, by the definition of $\Phi$, we have $\Phi(t,0)=(p(t),0,w(t))$ and, from the fact that $\gamma $ is parametrized by arc length, we have $\Vert \dot \gamma(t)\Vert=1$, which implies $(\dot p)^2+(\dot w)^2=0$ and hence $\dot p\ddot p+\dot w\ddot w=0$, that is  $\dot w\ddot w=-\dot p\ddot p$. Using this last result, we get
 \[
\left\{
\begin{array}{c}
p^2\dot p\ddot p=0\\
p\left[\dot w\dot p \ddot p-\ddot w(\dot p^2)\right]=0
\end{array}
\right.
\]
and we conclude that if $\gamma(t)=\Phi(t,0)$ is a geodesic, then the only possibility is that $\dot p=0$ and hence the price is constant.

Finally, observe that the coordinate curve $\Phi(0,t)$ is always a geodesic, because it is a solution of 
 Equation \eqref{sist} with  $ \gamma_0=0$ and $ \gamma_1=t$.

\section{The case  $M=2$}\label{mainsection}

Let us consider a more general case, which includes the previous one as a special case, that is,
an economy with two consumers and an arbitrary number of goods.
In this case  $E(r)$ is a manifold of dimension $L$
 in a space of dimension $2L-1$.    
An explicit parametrization of $E(r)$ is given by \cite{lmu}
 \begin{equation*}  
\begin{split} 
 \Phi:\: \mathbb{R}^L &  \overset{\:\:\:\:\:\:\:\:\:\:\:\:\:\:\:\:\:\:\:\:\:\:\:\:\:\:\:}{\longrightarrow} E(r)\\
        (t,\alpha_1,\dots,\alpha_{L-1} )& \longmapsto  (p_1(t),\dots,p_{L-1}(t),\alpha_1,\dots,\alpha_{L-1},w(t)-p_1(t)\alpha_1-\dots-p_{L-1}(t)\alpha_{L-1}),\\ 
\end{split} 
\end{equation*}
where we set $x_0=t,\:\:x_1=\alpha_1,\dots,x_{L-1}=\alpha_{L-1}$.\\

\noindent The following theorem shows that the finite geodesic property implies global uniqueness.

\vskip 0.5 cm

\begin{theorem}\label{mainteor} Let $\overrightarrow{\alpha}_j$ be the $L-1$-entry vector consisting of $L-2$ zeros and $1$ in the $j$-th position. If the $L$ coordinate curves $\Phi(t,0)$ and $\Phi(0,\overrightarrow{\alpha}_j)$ are geodesics, then the price is constant.

\noindent Proof: {\rm Suppose that the curve $\gamma(t)=\Phi(t,0,\dots,0)$ is a geodesic parametrized by arc length.
We have that  $\Vert \dot \gamma(t)\Vert=1$, so $\sqrt{B+A^2}=1$ and then  $C+AA'=0$  where\\

$\fbox{$A= \dot w-\dot p_1 x_1-\dots-\dot p_{L-1}x_{L-1}$}$         $\fbox{$B=(\dot p_1)^2+\dots +(\dot p_{L-1})^2$}$              $\fbox{$C=\dot p_1\ddot p_1+\dot p_2\ddot p_2+\dots+\dot p_{L-1}\ddot p_{L-1}$}$ \\

 $ \fbox{$\Vert p\Vert ^2=(1+p_1^2+\dots+p_{L-1}^2)$} $
      $ \fbox{$A'= \ddot w- \ddot p_1x_1-\dots- \ddot p_{L-1}x_{L-1}$}. $ \\

The Christoffel symbols are
 \[
\begin{split}
\Gamma_{00}^0=&\frac{[\Vert p\Vert^2C+AA']}{\Vert p \Vert ^2B+A^2}\\
\Gamma_{00}^k=&\frac{p_{k}\left[AC-A'B\right]}{\Vert p \Vert ^2B+A^2}\\
\Gamma_{0j}^0=&\frac{-\dot p_{j}A}{\Vert p \Vert ^2B+A^2}\\
\Gamma_{0j}^k=&\frac{\dot p_{j}p_{k}B}{\Vert p \Vert ^2B+A^2}\\
\Gamma^0_{ij}=&\Gamma^0_{ii}=0\\
\Gamma^k_{ij}=&\Gamma^k_{ii}=0 \end{split}\]

and the geodesic equation from Equation \eqref{geodesic} 
becomes 
{\small\[
\left\{
\begin{array}{c}
\ddot \gamma_0(t)+\Gamma^0_{00}\dot \gamma_0(t)^2+2\Gamma^0_{01}\dot \gamma_0(t)\dot \gamma_1(t)+2\Gamma^0_{02}\dot \gamma_0(t)\dot \gamma_2(t)+\dots+2\Gamma^0_{0(L-1)}\dot \gamma_0(t)\dot \gamma_{L-1}(t)+ \mbox{Nulls $\Gamma_{ij}$ }=0,\\
\ddot \gamma_1(t)+\Gamma^1_{00}\dot \gamma_0(t)^2+2\Gamma^1_{01}\dot \gamma_0(t)\dot \gamma_1(t)+2\Gamma^1_{02}\dot \gamma_0(t)\dot \gamma_2(t)+\dots+2\Gamma^0_{0(L-1)}\dot \gamma_0(t)\dot \gamma_{L-1}(t)+ \mbox{Nulls $\Gamma_{ij}$ }=0,\\
\dots\\
\ddot \gamma_{L-1}(t)+\Gamma^{L-1}_{00}\dot \gamma_0(t)^2+2\Gamma^{L-1}_{01}\dot \gamma_0(t)\dot \gamma_1(t)+2\Gamma^{L-1}_{02}\dot \gamma_0(s)\dot \gamma_2(t)+\dots+2\Gamma^{L-1}_{0(L-1)}\dot \gamma_0(t)\dot \gamma_{L-1}(t)+ \mbox{Nulls $\Gamma_{ij}$}=0.\\

\end{array}
\right.
\]}

If $\gamma(t)=\Phi(t,0,\dots,0)$, then the geodesic equations are
\[
\left\{
\begin{array}{c}
\Gamma^0_{00}=0\\
\Gamma^1_{00}=0\\
\dots\\
\Gamma^{L-1}_{00}=0,\\

\end{array}
\right.\:\:\:\:\mbox{   that is,                     }\:\:\:\:\:\:
\left\{
\begin{array}{c}
\frac{[\Vert p\Vert^2C+AA']}{\Vert p \Vert ^2B+A^2}=0\\
\frac{p_{1}\left[AC-A'B\right]}{\Vert p \Vert ^2B+A^2}=0\\
\dots\dots \dots\\
\frac{p_{L-1}\left[AC-A'B\right]}{\Vert p \Vert ^2B+A^2}=0\\
\end{array}
\right.
\]

which can be simplified as follows

 \[
\left\{
\begin{array}{c}
[\Vert p\Vert^2C+AA']=0\\
p_{1}\left[AC-A'B\right]=0\\
\dots\dots \dots\\
p_{L-1}\left[AC-A'B\right]=0\\
\end{array}
\right.
\]
\\

On the other end, if we suppose that  $-AA'=C$, we get
  \[
\left\{
\begin{array}{c}
-(p_1^2+p_2^2+\dots+p_{L-1}^2)AA'=0\\
p_{1}\left[-A^2A'-A'B\right]=0\\
\dots\dots \dots\\
p_{L-1}\left[-A^2A'-A'B\right]=0.\\
\end{array}
\right.
\]
This,  applied to our curve, gives 
\[
\left\{
\begin{array}{c}
-(p_1^2+p_2^2+\dots+p_{L-1}^2)AA'=0\\
-p_{1}\ddot w(A^2+B)=0\\
\dots\dots \dots\\
-p_{L-1}\ddot w(A^2+B)=0.\\
\end{array}
\right.
\]
which implies $\ddot w=0$.\\
Now, if we also assume that $\gamma(t)=\Phi(0,1,0,\dots,0)$ is a geodesic,
we obtain
\[
\left\{
\begin{array}{c}
-(p_1^2+p_2^2+\dots+p_{L-1}^2)AA'=0\\
-p_{1}(\ddot w-\ddot p_1)(A^2+B)=0\\
\dots\dots \dots\\
-p_{L-1}(\ddot w-\ddot p_1)(A^2+B)=0\\
\end{array}
\right.
\]
which implies $\ddot p_1=0$ and hence $\dot p_1=c$ and $p_1(t)=c\,t+d$.
But, since the $p_j$'s are positive, $p_1(t)$ must be constant. 
The same reasoning can be applied to show that every $p_j$ is constant, 
by letting $\gamma(t)=\Phi(0,0,\dots,1,\dots,0)$ with $1$ in the $j$-position.}
\qed
\end{theorem}

The following two remarks highlight the role played by the economic assumptions.

\begin{remark}\rm
The economic assumption of positive prices plays a crucial role as we can see from the following counterexample.
Consider the manifold with the following parametrization, where $p_1,\dots, p_{L-2}$ are constants and $p_{L-1}(t)=t$,
 \begin{equation*}  
\begin{split} 
 \Phi:\: \mathbb{R}^L &  \overset{\:\:\:\:\:\:\:\:\:\:\:\:\:\:\:\:\:\:\:\:\:\:\:\:\:\:\:}{\longrightarrow} \mathbb R^{L-1}\\
        (t,\alpha_1,\dots,\alpha_{L-1} )& \longmapsto  (p_1,\dots,t,\alpha_1,\dots,\alpha_{L-1},w(t)-p_1\alpha_1-\dots-t\alpha_{L-1}).\\ 
\end{split} 
\end{equation*} 
The curves $\Phi(t,0,\dots,0)$, $\Phi(0,1,\dots,0)$ and $\Phi(0,0,\dots,1,\dots, 0)$ are all geodesics of this manifold, but the price is not constant.
\end{remark}

\begin{remark}\rm
We provide an example of a ruled hypersurface $F(t,\alpha)\subset \mathbb{R}^{L+1}$, $t\in\mathbb{R}$ and $\alpha\in\mathbb{R}^{L-1}$, with zero sectional curvature that is not a hyperplane, even if
the set of functions $F(t, \bar \alpha)$, for every $\bar \alpha$ in $\mathbb{R}^{L-1}$, are geodesics.
Let $F(t, \alpha)=\gamma(t)+\alpha_1e_1+\dots+\alpha_{L-1}e_{L-1}$, where $e_1,\dots,e_{L-1}$ is the canonical orthonormal basis of $\mathbb{R}^{L-1}$ and $\gamma(t)=(0,\dots,0,\gamma_1(t),\gamma_2(t))\in \mathbb{R}^{L+1}$. Moreover, the tangent space of $F(t,\alpha)$ is given by the span of  $\partial F_t=(0,0,\dots,\dot \gamma_1(t),\dot \gamma_2(t))$ and $\partial F_{\alpha_j}=e_j=(0,0\dots,1,\dots,0)$ with $j=1,\dots,L-1$. Then the curves $F(t,0)$ and $F(t, \bar \alpha)$ are geodesics, in fact $\ddot F(t,0)=(0,\dots,0,\ddot \gamma_1(t),\ddot \gamma_2(t))$ and $\ddot F(t, \bar \alpha)=(0,\dots,0,\ddot \gamma_1(t),\ddot \gamma_2(t))$ are both orthonormal to the tangent space and so parallel to the normal vector of the hypersurface.
\end{remark}

In the following corollary we summarise what, to the best of our knowledge, is known about the relationship between curvature, minimal entropy, geodesics and uniqueness on smooth exchange economies.

\begin{corollary}\label{corequiv}
The following properties of the equilibrium manifold $E(r)$
are equivalent for an  economy with $M=2$ agents and an arbitrary number of goods:

\begin{itemize}
\item[(1)] constant equilibrium price for every economy;

\item[(2)] uniqueness of equilibrium for every economy;

\item[(3)] zero curvature property;

\item[(4)] minimal entropy property;

\item[(5)] the $L$ coordinate functions are geodesics.

\end{itemize}

\end{corollary}

\begin{proof}
First observe that $(1)\implies (2)$ (obvious) and  $(2)\implies (1)$ by \citep[p. 188 Theorem 7.3.9 part (2)]{balib}, hence $(1)\iff (2)$. Moreover, property $(1)$ means that the equilibrium correspondence is flat, so we have that $(1)\implies (3)$ and $(1)\implies (5)$. On the other hand, 
$(3)\implies (2)$  by \citep[Corollary (4)]{lmu} and $(5)\implies (2)$  by Theorem  \ref{mainteor}.
Finally, $(2)\implies (4)$ \citep[Remark 2.2]{lment} and $(4)\implies (2)$ by \citep[Theorem 3.1]{lment}.
\end{proof}

We conjecture that the equivalences in Corollary \ref{corequiv} hold for any number of goods or agents. What it is known is that the equivalence $1\iff 2$ holds for an arbitrary number of goods or agents.  Moreover, for two-good economies $L=2$ with an arbitrary number of agents, $(3)\iff (2)$ by \citep[Theorem 5.1]{lmcurvuni}. On the other hand, the known sufficient conditions used to extend the equivalence $(4)\iff (2)$ to the same setting are quite strong \citep[Theorem 4.1]{lment}.
The conjecture in the general case is still open.

\end{document}